\documentclass[pra,twocolumn,superscriptaddress,preprintnumbers,floatfix]{revtex4}
\usepackage{graphics}
\usepackage{amssymb}
\usepackage{color}
\usepackage{graphicx}
\usepackage{amsthm}
\usepackage{paralist}
\usepackage{verbatim}
\usepackage{chemarr}
\newtheorem{theorem}{Theorem}
\newtheorem{lemma}{Lemma}
\newtheorem{corollary}{Corollary}
\def\EQ#1{\begin{eqnarray}#1\end{eqnarray}}

\begin{document}

\title{Fluctuations in Single-Shot $\epsilon$-Deterministic Work Extraction}

\author{Sina Salek}
\email[]{salek.sina@gmail.com}
\affiliation{Department of Computer Science, The University of Hong Kong, Pokfulam Road, Hong Kong}
\affiliation{School of Mathematics\\
University of Bristol, Bristol BS8 1TW, United Kingdom}

\author{Karoline Wiesner}
\affiliation{School of Mathematics\\
University of Bristol, Bristol BS8 1TW, United Kingdom}
\affiliation{Bristol Centre for Complexity Sciences\\
University of Bristol, 1-9 Old Park Hill, Bristol BS2 8BB, United Kingdom}

\begin{abstract}
There has been an increasing interest in the quantification of nearly deterministic work extraction from a finite number of copies of microscopic particles in
finite time.  This paradigm, so called single-shot $\epsilon$-deterministic work extraction, considers processes with small failure probabilities. However, the resulting 
fluctuations in the extracted work entailed by this failure probability have not been studied before. In the standard thermodynamics paradigm fluctuation theorems are powerful tools to study fluctuating quantities.
Given that standard fluctuation theorems are inadequate for a single-shot scenario,
here we formulate and prove a fluctuation relation specific to the single-shot $\epsilon$-deterministic work extraction to bridge this gap. 
Our 
results are general in the sense that we allow the system to be in contact with the heat bath at all times. As a corollary of our theorem we derive the known 
bounds on the $\epsilon$-deterministic work.
\end{abstract}
\maketitle

\section{Introduction \label{intro}}

Thermodynamics has been historically developed as a discipline that deals with a large number of systems, e.g. an ensemble of particles in a box. In recent years, however, there has been significant interest in studying
the relationship between thermodynamics and statistical physics, for finite number of copies of particles and finite application of operations on them. This limit is called single-shot thermodynamics
\cite{DelRio2011,Gemmer2015,Horodecki2013,AAberg2013}.
A particularly interesting question in single-shot thermodynamics is extraction of work from systems that are out of equilibrium \cite{AAberg2013, Horodecki2013}. In order to formulate work extraction, we first need to
define work. The first law of thermodynamics splits energy into an ordered form and a disordered form. The ordered form of energy is called work. The challenge faced in the quantification of work in the single-shot
regime is that, for small ensembles of microscopic systems, fluctuations dominate. Hence, determining the typical behaviour of systems, for which one would need large ensembles, is challenging. Thus, it is not clear 
how one can quantify ordered energy, i.e. work as perceived in standard thermodynamics. A way around this is to design a process which uses an amount of energy to lift the state of an external system, called `battery'
from a single energy level to another single energy level. This way one can define a notion of work that can be quantified in the single-shot regime, and that also corresponds to the notion of work as understood in
standard thermodynamics. Nevertheless, since fully deterministic work extraction can be too stringent a constraint, one might also allow processes that fail with a small probability $\epsilon <<1$. 
Such a process is called `$\epsilon$-deterministic' which is to say `nearly deterministic'.
Horodecki and Oppenheim in \cite{Horodecki2013} discussed a scenario involving a system out of thermal equilibrium and block diagonal in the energy eigenbasis, coupled to a bath and battery. The task is to lift the 
state of the battery from the ground state to its excited state $\epsilon$-deterministically. If the task is done successfully, a certain amount of work is stored in the battery. 
In this scenario of a small failure probability $\epsilon <<1$ they find an upper bound 
for the work extracted.
The existence of a failure probability entails some fluctuations in the work extracted. Experiments which fail to extract exactly the desired amount of work are dismissed entirely. What has not been considered so far 
is that rather than dismissing these experiments one can quantify the alternative amount of work that has been extracted and how it fluctuates. Prior to the work presented here it was not known how to characterise
these fluctuations. 
In this paper, we formulate and prove a theorem (Theorem \ref{epsilon}) that characterises the fluctuations of the extractable work due to the failure probabilities in single shot thermodynamics. 

To answer this, we use an idea developed in another domain of statistical 
physics.
In statistical physics powerful tools called fluctuation theorems \cite{Campisi2011}
have been developed to characterise fluctuations in quantities such as work or entropy. They put restrictions on the probability density function of
these fluctuating quantities. 
In the case of work, for instance, these theorems compare the probability of work cost of a thermodynamic process with the probability of 
work gain of its reverse process. 
An example of fluctuation theorem is Crooks work relation \cite{Crooks1999}. Crooks theorem relates the probability, $P(w,\mathcal{P})$, of work cost, $w$, of 
a process to the probability, $P(-w,\mathcal{P}^{rev})$, of work gain in the time-reverse of that process, where $\mathcal{P}$ and $\mathcal{P}^{rev}$
are the process and its inverse respectively. The relation is formally expressed as $\frac{P(w,\mathcal{P})}{P(-w,\mathcal{P}^{rev})}=e^{\beta (w- \Delta F)}$,
where $\beta$ is the inverse temperature of the environment divided by the Boltzmann constant, and $\Delta F$ is the equilibrium free energy difference of 
the initial and final state. A thermodynamic process, there, is defined by changing the Hamiltonian of the system at an arbitrary speed, according to a
specific trajectory of the Hamiltonian. And the reverse process is one where the Hamiltonian is brought back to the initial setting, according to the reversed
trajectory of the forward process. The ratio of these two probabilities equals the exponential of the work dissipated in the process, $w_{diss}=w-\Delta F$. Therefore, in a process
where no work is dissipated, \emph{e.g.} an isothermal expansion, these two probabilities are the same and no fluctuation in work occurs. As we see, this theorem
characterises fluctuations in a thermodynamical quantity that meets the conditions of the theorem, by comparing the probability of the forward process to that
of the backward process. In this paper we show that the main ideas of fluctuation relations can be used to characterise the fluctuations arising from the $\epsilon$
failure probabilities. However, since the currently existing fluctuation theorems have different assumptions to the ones in the single-shot work extraction, we
formulate our theorem specifically for the single-shot scenario. This theorem will put the notion of fluctuations in $\epsilon$-deterministic work extraction, 
and that in standard thermodynamics on an equal footing. 

To formulate a fluctuation theorem for single shot thermodynamics, we relax the original assumption of Crooks work relation that the system, both in the forward process and its inverse, starts in a state
of thermal equilibrium with its environment. Our setup, in keeping with the setting in \cite{Horodecki2013}, consists of a work system, a
battery, and a thermal bath. The work system
remains in contact with the bath throughout the process. 
The protocol extracts work $W$ by performing some operations on the work system and as a result lifts the battery to the excited state 
$|W\rangle\langle W |$ with $\epsilon << 1$ failure probability. 
As a corollary of our
result we obtain the bounds of the extractable work found by \AA berg
\cite{AAberg2013}.

It is worth noticing that as a byproduct of choosing the present setting, we avoid a potential complication in the fluctuation analysis of open 
systems. In particular, thermodynamic work, $W$, is measured in a number of different ways in the context of fluctuation 
relations. Predominantly, the so-called scheme of Two-Measurement Protocol (TMP) is used were the eigenvalues of the initial and final Hamiltonians are measured,
and work, $W$, is defined as the difference between the two eigenvalues. Considering the first law of thermodynamics, we observe that for closed systems with 
no heat flow this quantity is the same as total work. However, in the case of open systems with strong coupling TMP does not provide a good measure of work. 
In our setup all of the work goes to lifting the state of the battery and despite the fact that the heat bath is attached to the work system at all times, 
a simple act of measuring the initial and final eigenvalues of the battery gives a correct measure of the work.

In the next section we give the technical preliminaries before presenting the main theorem (Theorem~\ref{epsilon}) in Section \ref{smooth}. The proof of Theorem~\ref{epsilon} is given in the 
Appendix.

\section{Preliminaries \label{prelim}}

The tool with which we choose to prove the theorem is the so-called resource theory of states out of thermal equilibrium. Resource theories \cite{Gour2008,Brandao2013} 
are mathematical formalisms, designed to answer questions such as `what state can be prepared in a given 
physical situation?' or `is a transition from one state to another possible, and if yes what is the transition probability?'. By answering these questions, 
resource theories allow us to study what happens to the state of a system in a physical situation in full generality.
Resource theories achieve this by restricting the allowed operations to a predefined set. A given set of operations one can prepare a set of states for free.
Any state outside that set is considered a resource. For instance, it is a well-known fact that if we restrict ourselves to local operations and classical 
communications (LOCC), we can create separable states. For LOCC, therefore, an entangled state is considered a resource.
The resource theory of athermal states is defined by restricting our allowed operations to the so-called thermal 
operations. Thermal operations, as explained concisely in \cite{Alhambra2016}, are the set of maps characterised by the following rules $i)$ a system with 
any Hamiltonian in the Gibbs state of that Hamiltonian can be added, $ii)$ any subsystem can be discarded through tracing out and $iii)$ any energy-conserving
unitary, i.e. those unitaries that commute with the total Hamiltonian, can be applied to the global system. For these operations a state out of equilibrium is
a resource and can be utilised to extract work. The commutation condition is imposed to ensure that total energy is conserved. As we see, the resource theory of
athermal states provides a framework wherein one can study the most general thermodynamical state transformations that respect energy conservation. Specifically,
in our theorem, as we are interested in thermodynamical processes that extract a certain amount of work, we use this framework to identify 
all possible processes that meet our conditions.

In our setup, initially the total 
system consists of the work system in state $\rho_s$ from which we intend to extract work, a bath in state $\tau_{\text bath}$ and a battery in its pure ground 
state. The 
battery is a finite-dimensional many-level system, used for storing the ordered energy, i.e. work. A many-level battery is needed in our case 
in order to store the different possible amounts of the fluctuating work. 
A thermal operation is used to extract work from the work state. The thermal operation is generated by a global unitary operator on the initial state of 
the total system. This operator creates correlations between the work state, bath, and battery. As we can accept a small failure probability, the final
state of the system can be slightly different every time we run the protocol. In particular the battery may be charged to the point $W^\delta$, where $\delta$
here signifies different excited levels of the battery. The thermal operation is chosen such that it aims to transform the initial state of the work system
to the Gibbs state $\tau_S$. However, again, due to the failure probability, the work system might end in a state that is not exactly Gibbs, we call this state $\tau_S^\delta$.
A more detailed study of the final state in the setting we choose can be found in \cite{Gemmer2015}. To formulate our fluctuation theorem we also need to 
specify a backward process. The backward process is generated by the unitary operator, inverse of the global unitary used in the forward process. We choose the 
initial state of the backward process to be the correlated state of the work system and bath, such that the work state is left in a Gibbs state, after
tracing out the state of the bath, as explained below. The initial state of the battery in the backward process is the same as its final state in the forward process. This guarantees
that the work extracted in the forward process and the work cost in the backward process are the same. The forward and backward processes are illustrated in 
Fig. (\ref{states}).

Our choice of the initial state of the backward process is due to a subtlety inherent in single-shot thermodynamics. In \cite{Brandao2013} it was shown that 
there is an intrinsic irreversibility in the single shot scenario in the sense that if a resource $\rho$ can be converted to a resource $\sigma$, in general 
one cannot assume that the resource $\sigma$ can be converted back to the initial resource $\rho$ under the same conditions. In particular, in the first part 
of \cite{Horodecki2013} the maximum work extractable in the work extraction setting above was found, where work was extracted by transforming the resource 
state $\rho$ into the Gibbs state. In the second part, a reverse of that process was defined as a task of starting with a bath, a resource system in Gibbs state and a
battery, all initially in tensor product with each other. It was shown that the work cost of returning the Gibbs state back to $\rho$ is greater than the work
extracted in the first part in the single-shot scenario. This is due to the detailed correlations that are created between the system and the bath after the 
application of the global unitary that generates the thermal operations, which are significant for single-shot work extraction. Therefore, the backward process
for a fluctuation relation described above has to begin in a correlated state of resource system and the bath, such that the correlations are produced by the 
operations of the kind used in the forward process. 

\begin{figure}[t!]
  \centering
      \includegraphics[width=0.45\textwidth]{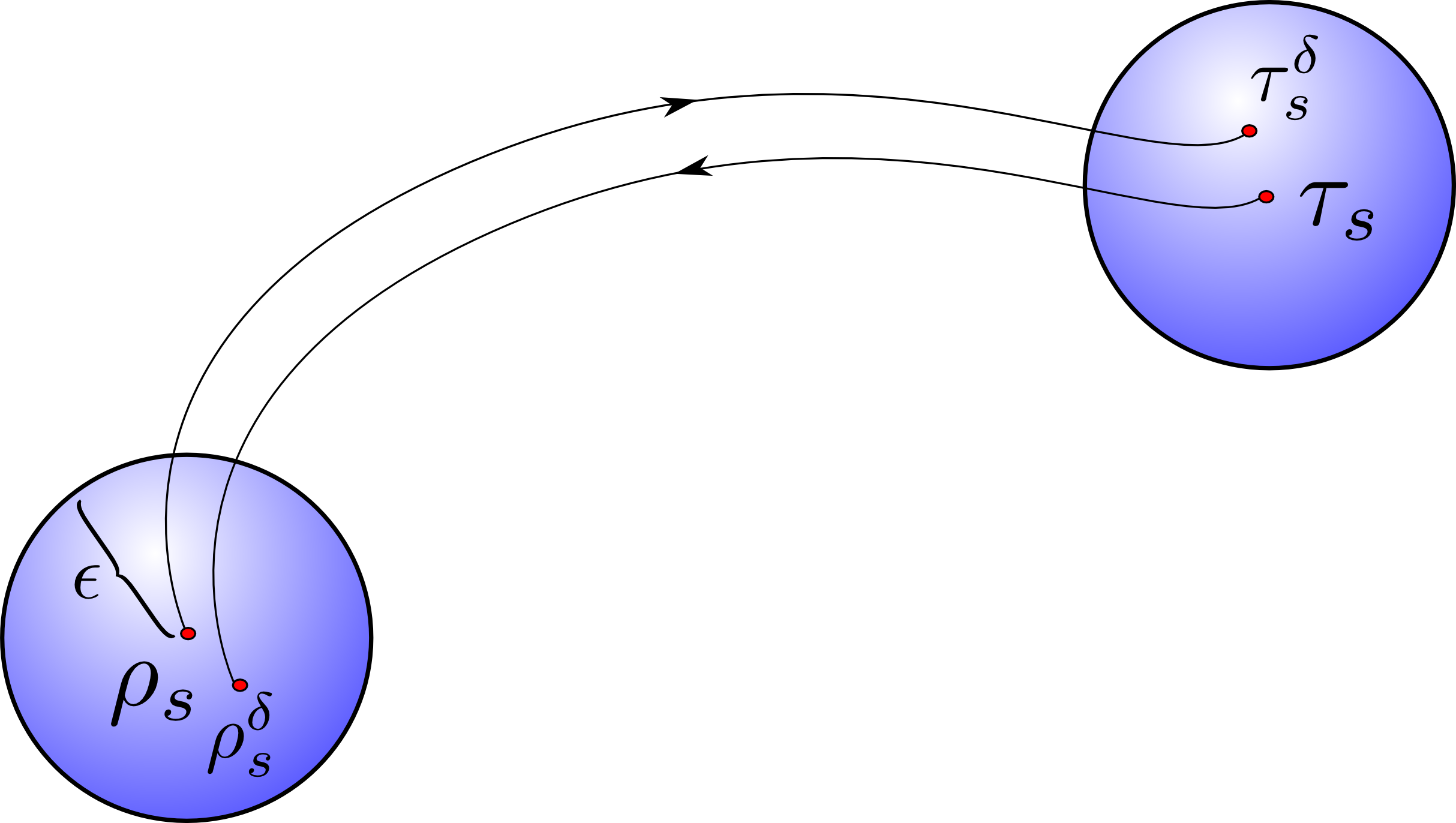}
  \caption{\small 
  In the forward process, the system is initially prepared in the state $\rho_s$. Evolving the system and the bath together according to some global unitary operator
  correlates the state of the two. 
  The transformation is such that a work $W^\delta$ is extracted and the final state of the system $\tau_s^\delta$ is close to Gibbs state. The backward process
  is defined by the dual operator to the unitary in the forward process, and the initial state of the bath and the system are correlated, such that tracing out
  the state of the bath leaves the system in the Gibbs state. The final state of the backward process is in the ball of $\epsilon$-close states to the initial 
  state $\rho_s$.
  }
  \label{states}
\end{figure}

In our theorem we use the $\epsilon$-smooth Renyi relative entropy of order $\alpha = 0$.
The Renyi relative entropy of order $0$ is defined as $D_0(\rho_s || \tau_s):=- \ln Tr[\rho_s^0 \tau_s]$, with $\rho_s^0$ being the support of $\rho_s$. The 
$\epsilon$-smooth version of this entropy is defined by maximising it over the distributions that are $\epsilon$-close to the state of the system, formally defined as 
\EQ{
D_0^\epsilon(\rho_s || \tau_s):=\max_{\bar{\rho_s} \in \mathcal{B}^\epsilon(\rho_s)} D_0(\bar{\rho_s} || \tau_s). \label{diver}
}
with
\EQ{
 \mathcal{B}^\epsilon(\rho_s):=\{ \bar{\rho_s} \geq 0 : || \bar{\rho_s}-\rho_s||_1 \leq \epsilon, Tr(\bar{\rho_s}) \leq Tr(\rho_s)\}
}
defining the ball of $\epsilon$-close distributions. This entropy was shown to act as a non-equilibrium free energy difference in the single-shot
regime \cite{Horodecki2013, AAberg2013}. 

\section{The Main Theorem \label{smooth}}

We now state the main result, a fluctuation theorem relating the ratio between the probability of extracting 
work $W$ in an $\epsilon$-deterministic forward process and the probability of putting work $W$ back into the system in a reversed 
process.

\begin{theorem}
 \label{epsilon}
 Consider a total system $\tau_{bath}\otimes \rho_s \otimes |0\rangle \langle 0|$
consisting of a work system in some state $\rho_s$, diagonal in energy eigenbasis, a thermal bath at temperature $T$ in state
$\tau_{\text  bath}$, and a many-level  battery system in its pure ground
state $|0\rangle \langle 0|$. Furthermore, assume $\epsilon$-deterministic
forward process
$\mathcal{P}^\epsilon$, defined by a global unitary on the total system that commutes with its Hamiltonian, and a reverse process $\mathcal{P}^{\epsilon,rev}$,
defined by inversing the global unitary of the forward process. The initial state of the backward process is a correlated state of the bath and 
the resource system, such that the work state is in a Gibbs state, $\tau_S$, after tracing out the bath. With these assumptions,
the ratio of work extraction 
probability, $P(W,\mathcal{P}^\epsilon)$, to work cost probability, $P(-W,\mathcal{P}^{\epsilon , rev})$, is given by
\EQ{
\frac{P(W,\mathcal{P^\epsilon})}{P(-W,\mathcal{P}^{\epsilon,rev})} = e^{\beta W +\ln (1-\delta) - D_0^\epsilon (\rho_s||\tau_s)}, \label{theorem}
}
where $W$ is the work extracted. $\beta = 1/kT$, where $k$ is the Boltzmann constant and $T$ is the
temperature. $D_0^\epsilon (\rho_s||\tau_s)$ is the smooth version of Renyi relative entropy of order $\alpha=0$. $\ln(1-\delta)$ is the amount 
of deviation from the maximum extractable work, with $0 \leq \delta \leq \epsilon$.
\end{theorem}

We provide the proof of Theorem~\ref{epsilon} in the Appendix.
Theorem \ref{epsilon} characterises the fluctuations in the extracted work due to the failure probability in the single-shot $\epsilon$-deterministic work extraction scenario. It
states that if a quantum resource is used to extract work in an $\epsilon$-deterministic process by thermal operations,
the probability of the work extracted is related to the probability of the work put back into the system in its reversed process by an exponential factor.
The term in the 
exponent of Eq. (\ref{theorem}) is the analogue of dissipated work. This characterises a type of irreversibility in the
finite-run behaviour of a microscopic system in the single-shot regime, akin to that of Crooks relation. As discussed before, the final state of the system 
in the forward process may be such that the state of the system is not exactly Gibbs. The term $\ln(1-\delta)$ signifies this deviation. We show the details
of this in the Appendix.

In the following corollary we shall observe how a known bound of work extraction follows from our result.

\begin{corollary} \label{cor}
The work that can be extracted from a system out of equilibrium by thermal operations in an $\epsilon$-deterministic process is bounded by 
\begin{widetext}
 \EQ{\label{eq.cor}
 kTD_0^\epsilon (\rho_s||\tau_s)+kT\ln (1-\epsilon)-kT\ln(1-\delta) \leq W \leq kTD_0^\epsilon (\rho_s||\tau_s)-kT\ln (1-\epsilon)-kT\ln(1-\delta),
 } with $0 \leq \delta \leq \epsilon$.
\end{widetext}
\end{corollary}

\begin{proof} To see this, we first derive the lower bound and then the upper bound as follows. For the lower bound we multiply both sides of the Eq. (\ref{theorem}) by $P(-W,\mathcal{P}^{\epsilon,rev})$. From the fact that 
$P(W,\mathcal{P^\epsilon})\geq 1-\epsilon$
we have
\EQ{
1-\epsilon &\leq& P(-W,\mathcal{P}^{\epsilon,rev}) e^{\beta W+\ln (1-\delta) - D_0^\epsilon (\rho_s||\tau_s)}\\
&\leq& e^{\beta W+\ln (1-\delta) - D_0^\epsilon (\rho_s||\tau_s)}, \label{05alcor}
}
where in the second inequality we used the fact that any probability is less than or equal to one.
Taking the logarithm of Eq. (\ref{05alcor}), the lower bound in the statement of Corollary~\ref{cor} follows.

For the upper bound we rewrite Eq. (\ref{theorem}) as 
\EQ{
\frac{P(-W,\mathcal{P}^{\epsilon,rev})}{P(W,\mathcal{P^\epsilon})} = e^{-\beta W -\ln (1-\delta) + D_0^\epsilon (\rho_s||\tau_s)}. 
}
Following the same procedure as above, here we have 
\EQ{
1-\epsilon &\leq& P(-W,\mathcal{P}^{\epsilon}) e^{-\beta W-\ln (1-\delta) + D_0^\epsilon (\rho_s||\tau_s)}\\
&\leq&  e^{-\beta W-\ln (1-\delta) + D_0^\epsilon(\rho_s||\tau_s)}, \label{05alcor2}
}
Taking the logarithm gives the upper bound of Corollary~\ref{cor}.
\end{proof}

In \cite{AAberg2013}, \AA berg gives the bounds on work extraction in $\epsilon$-deterministic work extraction as 
\EQ{
kTD_0^\epsilon (\rho_s||\tau_s) \leq W \leq kTD_0^\epsilon (\rho_s||\tau_s) -kT\ln (1-\epsilon). \label{AAbergcor}
}
Notice that the upper bound in Eq. ({\ref{AAbergcor}}) is the special case of Eq.~\ref{eq.cor} with $\delta = 0$
. The lower bound is the special case of Eq.~\ref{eq.cor} with $\delta =\epsilon$.

\section{Discussion \label{concl}}
In the discussion section of Ref. \cite{AAberg2013}, \AA berg 
leaves the question of the link between the fluctuations of the $\epsilon$-deterministic work extraction paradigm to that of the fluctuation theorems
paradigm open.
Theorem \ref{epsilon} above has answered this question.
One can see, as we show in the appendix, that in the absence of any fluctuations of the former type (the case where $\epsilon$ and $\delta$ are zero) the
the dissipated work goes to zero and the forward and backward probabilities take the same value. This means  
there will be no fluctuations of the latter type. It is only through the introduction of the failure probabilities that we obtain a fluctuation theorem,
and hence demonstrating the direct connection between the two. We have derived an equality in the same way Crooks
relation does for work probability density function in the standard account of thermodynamics. 
The relation in Theorem 
\ref{epsilon} quantifies the ratio of the probability of work extraction in a process to that of work cost in its backward process, in terms of
the smooth 
Renyi relative entropy of order $\alpha = 0$. As a corollary we find the known bounds of work extraction protocols. The necessity of formulating
and proving this variation of fluctuation theorem is that the single-shot regime has different assumptions to the standard thermodynamics, which renders the 
standard fluctuation relations inadequate for our purposes. In the standard fluctuation relations the work is determined by making several energy measurements
at the beginning and the end of the processes and comparing the energy eigenvalues. This requires many measurements. In the single-shot scenario we require 
to be able to determine work by a single measurement. In our setting this is done by letting the process store all of the work in the battery. Then the work 
is measured by a single energy measurement on the battery state.
Given the record in the experimental applications of the present fluctuation theorems \cite{Liphardt2002, Collin2005, Douarche2005}, we believe
our work will contribute to the experimental results on 
single-shot thermodynamics by providing a method to measure smooth Renyi entropies. As an outlook on further theoretical extensions of this work, one can think of
generalisations or restrictions on the allowed operations, as well as consideration of quantumness. For instance Faist \emph{et. al.} have shown that 
Gibbs-preserving maps outperform thermal operations
in the quantum regime \cite{Faist2015}. A further investigation into the hierarchy of operations and possible exploitation of quantum correlations and 
catalysts merits an extended discussion. Furthermore, it would be interesting to apply a version of Asymptotic Equipartition Theorem \cite{Tomamichel2009a}
to investigate the relationship between our theorem and the ones in the standard statistical mechanics. 
In this paper we have focused on the notion of $\epsilon$-deterministic work. However, probabilistic work in the single-shot scenario is also an 
interesting concept and has been recently studied \cite{Alhambra2016}.

\noindent{\bf Note added:} Similar results were obtained independently by Dahlsten et al, using a different set-up and different starting 
assumptions, in {\em Equality for worst-case work at any protocol speed}.

{\bf Acknowledgements}
We would like to thank Johan \AA berg, Jens Eisert, Oscar Dahlsten, Rodrigo Gallego, Henrik Wilming and Albert Werner for their comments and discussions. 
Part of this work was supported by the COST Action MP1209 "Thermodynamics in the quantum regime", Hong Kong Research Grant 
Council through Grant No. 17326616, National Science Foundation of China through Grant No. 11675136, HKU Postdoctoral Fellowship. K.W. acknowledges funding through EPSRC (EP/E501214/1).

{\small

\bibliography{SSFT}
\bibliographystyle{plain}

}

\appendix*
\section{Proof of the Theorem \ref{epsilon} \label{app}}
Here, we prove Theorem \ref{epsilon} by deriving Eq. (\ref{theorem}). However, to begin with we consider the simple case of a process which extracts the amount of work
$W$ from the work system with probability one, i.e. with failure probability $\epsilon = 0$. 
For such a case, the following lemma holds. This does not give a fluctuation
theorem directly, as without a failure probability there is no fluctuation. However, it is an important building block in proving the theorem. In this lemma we use Renyi relative entropy of order $\alpha =0$ defined
as $D_0(\rho_s||\tau_s):=-\ln Tr[\rho_s^0 \tau_s]$, with $\rho^0_s$ being the 
support of the initial state and $\tau_s$ the thermal state of the system. 
In the following lemma we adopt the same sign convention for work as used in \cite{Horodecki2013, AAberg2013}. Here we have the same assumptions as in 
\cite{Horodecki2013} for the allowed operations to be thermal operations, which are obtained by performing some global unitaries on the total initial state. 
In the following lemma the reverse process is obtained from the inverse of the global unitary in the forward process. The 
initial state of the reverse process, here, is set to be the same as the final state of the forward process. Our assumptions on the bath are also the same 
as \cite{Horodecki2013} namely that
\begin{inparaenum}[(i)]

\item the spectrum of the heat bath is continuous, i.e. for an energy of the heat bath $E_R$ and two arbitrary energies of the system, $E_S$ and $E'_S$, there 
exists $E'_R$, such that $E_R + E_S = E'_R + E'_S$, 

\item and around the energy $E$ the degeneracies can be written as 
\EQ{
\nonumber g(E+\Delta E)&=&e^{S(E+\Delta E)}\\
&=&e^{S(E)+\Delta E\frac{\partial S(E)}{\partial E}+O(\Delta E^2)}\label{approx} \\
&=& g(E)e^{\beta \Delta E+O(\Delta E^2)} \label{bath},
}
where $S(E):=\ln g(E)$ and $\beta :=\frac{\partial S(E)}{\partial E}$. Notice that for a large enough bath the second and higher order terms in Eq. 
(\ref{approx}) and (\ref{bath}) can be neglected, which is the assumption we adopt throughout this paper.
\end{inparaenum}

\begin{lemma}
Consider a total system $\tau_{bath}\otimes \rho_s \otimes |0\rangle \langle 0|$ consisting of a work system in some arbitrary state $\rho_s$, diagonal in energy
eigenbasis, subject to a 
time-independent Hamiltonian, a thermal bath at temperature $T$ in state $\tau_{\text  bath}$, and  a two-level  battery system in its pure ground state 
$|0\rangle \langle 0|$. Here, the work is extracted by lifting the state of the battery from its ground state to the pure excited state $|W\rangle \langle W|$.
The gap between the ground and the excited state of the battery is set to $W$.
Furthermore, assume processes $\mathcal{P}$ and $\mathcal{P}^{rev}$ consisting of Thermal Operations only. Then
\EQ{
\frac{P(W,\mathcal{P})}{P(-W,\mathcal{P}^{rev})} = e^{\beta
W- D_0(\rho_s||\tau_s)}~,
}
with $\beta= 1/kT$ where $k$ is the Boltzmann constant.   

\label{equality}
\end{lemma}

\begin{proof}
In \cite{Horodecki2013} it was shown that for a fixed total energy $E$ if the bath is very large compared to the system, the
state $\tau_{bath}\otimes \rho_s \otimes |0\rangle \langle 0|$ can be
written as $\bigoplus_{E_s} \eta_{E-E_s}^R \otimes \Pi_{E_s} \rho_s \Pi_{E_s} \otimes |0\rangle \langle 0|$, 
where $\eta_{E-E_s}^R:=\frac{\mathbb{I}^{E-E_s}_R}{g_R(E-E_s)}$, $g_R(E-E_s)$
is the degeneracy of the bath, and the identity $\mathbb{I}^{E-E_s}_R$ 
acts on a $g_R(E-E_s)$-dimensional space. Similarly, the final state can be written as 
$\bigoplus_{E^f_s} \eta_{E-E^f_s-W}^R \otimes \Pi_{E_s} \tau_s \Pi_{E_s}
\otimes |W\rangle \langle W|$. Since the total system is block diagonal in energy eigenbasis, we can treat 
the eigenvalues of it as probabilities. Denote the energy levels of initial and final state of the work system by $E_i$ and $E_j$,
respectively. The eigenvalues of the initial and final
state are $\frac{P(E_i)}{g(E-E_i)}$ and $ \frac{P(E_j)}{g(E-E_j-W)}$, respectively. 
Denote the transition current, \emph{i.e.} the number of eigenstates 
that go from $E_i$ to $E_j$ by $k_{i \rightarrow j}$. Then occupation probabilities of the final state of the system are given by 
\EQ{
P(E_j)=\sum_i k_{i \rightarrow j} \frac{P(E_i)}{g(E-E_i)}.
}
Each summand on the RHS is the joint probability of the system initially occupying the 
$i$th energy level and ending up in the $j$th energy level in the final state.
Since the protocol stores the energy difference between initial and final
state, i.e. the work $W$, in the
battery state the probability of a value of work $P(W,{\mathcal P})$ can be calculated by summing the
RHS over $j$.
Notice that the total number of eigenstates in initial energy level $E_i$ is

\EQ{\label{eq.kij}
d_i=\sum_j k_{i \rightarrow j}=g(E-E_i),
}
and that the total number of eigenstates in  final energy level $E_j$ is 
\EQ{
d_j=\sum_i k_{i \rightarrow j}=g(E-E_j-W).
}

We also need to consider the backward process. 
Recall
that the global unitary for the forward process gives transition currents $k_{i \rightarrow j}$ as the rate at which energy eigenstates are 
transformed from $E_i$ to $E_j$. 
Therefore, as we are restricting ourselves to the diagonal case, this current can always be used to directly characterise the backward 
transition currents,  
$k^{rev}_{j \rightarrow i}$, as the rate at which the energy eigenstates are transformed back from $E_j$ to $E_i$. 
Notice, for the reversed process we have
\EQ{\label{eq.kijrev}
\sum_i k^{rev}_{j \rightarrow i}=g(E-E_j-W)~, 
}
which is the number of eigenstates in the final state of the forward process. In the setting of this lemma, this is also the number of eigenstates 
of the initial state of the backward process.

As in the standard Crooks relation, it is also important to fix the initial state of the backward process. 
For this lemma we use the final state of the forward process as the initial state of the backward process. Here, the final state of the bath and the work system is
a correlated
state such that the system is in thermal state after tracing out the bath.

\emph{Remark 1 --} Formally speaking, a thermal operation $T$ is generated by a global unitary $U$, acting jointly on the reservoir, system and battery, such that 
$\sigma_{sb}:=T(\rho_{sb})=Tr_R(U \rho_{Rsb}U^\dagger)$, where $\rho_{Rsb}$ is the total state of the reservoir, system and battery, with a reduced state $\rho_{sb}$. The restriction on this global unitary is that it
has to commute with the total Hamiltonian of the reservoir, system and battery. In the current case, the global unitary sends the initial
state $\tau_{bath}\otimes \rho_s \otimes |0\rangle \langle 0|$ to the final state 
\EQ{
U\tau_{bath}\otimes \rho_s \otimes |0\rangle \langle 0|U^\dagger:=\sigma_{\{bath\}s}\otimes |W\rangle \langle W|,
}
such that the reduced state of the system, $\sigma_s$ is the Gibbs state.
These are the same initial and final states considered in the exact transformation 
case in \cite{Horodecki2013}. Indeed, in \cite{Horodecki2013} it was shown that a reverse process starting with a fresh bath will cost more energy than the forward process. However, defining the backward process
as the one produced by the inverse of the forward unitary, \emph{i.e.} $U^\dagger$, and starting from the final state of the forward process, clearly completely reverses the process as 
$U^\dagger U \tau_{bath}\otimes \rho_s \otimes |0\rangle \langle 0| U^\dagger U=\tau_{bath}\otimes \rho_s \otimes |0\rangle \langle 0|$. We shall make a remark on the backward process with non-zero epsilon later in this appendix, after we provide the proof of the theorem.

Below, we calculate the ratio of the work cost probability to that of work gain probability for such a process. 

\EQ{
\nonumber \frac{P(W,\mathcal{P})}{P(-W,\mathcal{P}^{rev})}&=&
\frac{\sum_{i,j}k_{i \rightarrow j} \frac{P(E_i)}{g(E-E_i)}}{\sum_{i,j}k^{rev}_{j \rightarrow i} \frac{P(E_j)}{g(E-E_j-W)}}\\
\nonumber &=& \frac{\sum_{i}g(E-E_i) \frac{P(E_i)}{g(E-E_i)}}{\sum_{j} g(E-E_j-W) \frac{P(E_j)}{g(E-E_j-W)}}\\
&=& \frac{\sum_i\sum_j k^{rev}_{j \rightarrow i}\frac{P(E_j)}{g(E-E_j-W)}}{\sum_{j}  P(E_j)}\label{05app1}
}

Using the fact that the system ends up
in the Gibbs state, together with the assumption on the bath in Eq. (\ref{bath}) that $g(E-E_j) = g(E)e^{-\beta E_j}$, one can rewrite the term
in the numerator as follows
\EQ{
\nonumber \sum_i\sum_j k^{rev}_{j \rightarrow i} \frac{P(E_j)}{g(E-E_j-W)}&=&\sum_i\sum_j k^{rev}_{j \rightarrow i}\frac{1}{g(E-W)Z}\\ 
&=&\sum_i \frac{g(E-E_i)}{g(E-W)Z} \label{05app2},
}
where $Z$ is the partition function of the Gibbs state. 

Inserting Eq. (\ref{05app2}) into Eq. (\ref{05app1}) and simplifying the fraction we find
\EQ{
\frac{P(W,\mathcal{P})}{P(-W,\mathcal{P}^{rev})}=\frac{e^{\beta W}\sum_{i\in supp(\rho_s)}e^{-E_i}}{\sum_{j} e^{-E_j}}, \label{simp}
}
where $supp(\rho_s)$ is the support of the state $\rho_s$.
Noticing that 
\EQ{
\nonumber - D_0(\rho_s||\tau_s)&=&\ln \sum_{i\in supp(\rho_s)}  e^{-\beta E_i} \\
&&- \ln \sum_{j\in supp(\tau_s)}  e^{-\beta E_j},
}
\EQ{
 \frac{P(W,\mathcal{P})}{P(-W,\mathcal{P}^{rev})} = \exp \{\beta W - D_0(\rho_s||\tau_s)\}. \label{fin}
}
\end{proof}

Notice that each of the probabilities in the numerator and the denominator of the fraction in the LHS of Eq. (\ref{fin}) equal one, i.e.
\EQ{
\nonumber P(W,\mathcal{P})&=&\sum_{i,j}k_{i \rightarrow j} \frac{P(E_i)}{g(E-E_i)} \\
\nonumber &=&\sum_{i}g(E-E_i) \frac{P(E_i)}{g(E-E_i)}\\
\nonumber &=&1,
}
and 
\EQ{
\nonumber P(-W,\mathcal{P}^{rev})&=&\sum_{i,j}k^{rev}_{j \rightarrow i} \frac{P(E_j)}{g(E-E_j-W)}\\
\nonumber &=&\sum_{j}g(E-E_j-W)\frac{P(E_j)}{g(E-E_j-W)}\\
\nonumber &=&1
}
since the sum of probabilities of the initial and final energy levels have to equal one.
Therefore, the fully deterministic work extraction has only a trivial distribution of one value of extracted
work and zero everywhere else. For a fluctuation theorem to say something about the work probability, it has to have a width.
Below we extend the relation to the $\epsilon$-deterministic case, thereby giving the proof of the Theorem.
To achieve this, we use the technique of
$\beta$-ordering as introduced in Supplementary Note 4 of \cite{Horodecki2013}
and as illustrated in
Fig. (\ref{beta}). The graph plots the Gibbs rescaled probabilities $P(E_i) e^{\beta E_i}$ against the truncated partition function $\sum_i e^{-\beta E_i}$ 
for each energy level of the work system, $E_i$, arranging states in decreasing order of their 
value $P(E_i) e^{\beta E_i}$. In this staircase function, the area of each rectangle is equal to the probability of the state, $P(E_i)$, and the total width 
is equal to $Z$, the partition function. 
Removing the maximum number of bins in Fig. (\ref{beta}), with the minimum total probability then corresponds to removing states from the right side of the $\beta$-ordered 
spectrum. This procedure is called \emph{smoothing}. 

\begin{figure}[t!]
  \centering
      \includegraphics[width=0.45\textwidth]{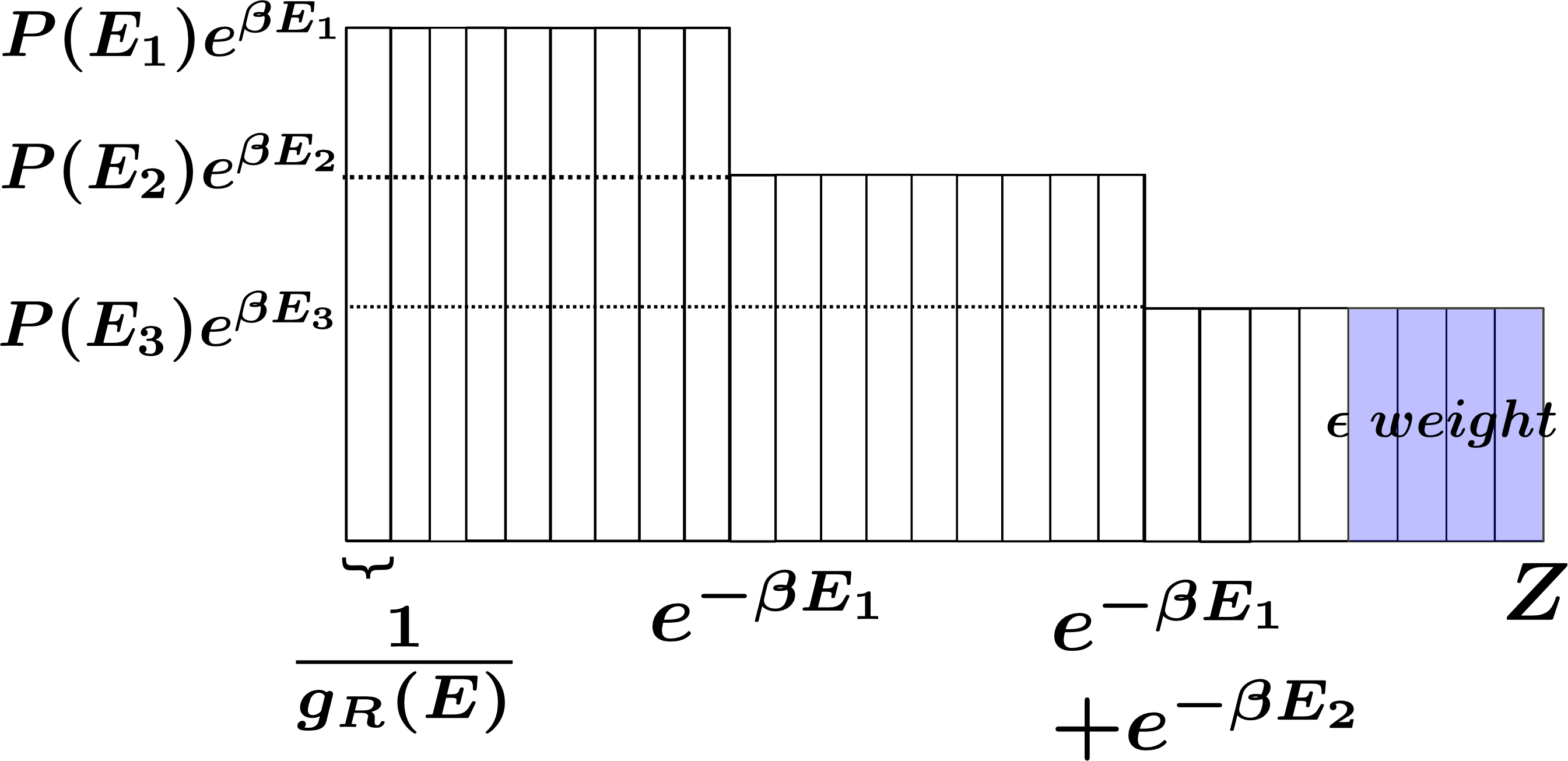}
  \caption{\small A weight $\epsilon$ is being taken out after $\beta$-ordering of the initial state. $\beta$-ordering refers to rearranging the eigenstates'
energies in decreasing order of their weight $P(E_i)e^{\beta E_i}$. Since we are interested in removing as many eigenstates as possible
while keeping a total weight of at least $1-\epsilon$, we cut from the far right end of the $\beta$-ordered spectrum. The white area shows the state that is 
$\epsilon$-close to the initial state and in our case maximises the Renyi relative entropy of order 0 in Eq. (\ref{diver}).}
  \label{beta}
\end{figure}

\begin{proof}

We want to express the LHS of Eq. (\ref{theorem}) in terms of the work extracted from the system. We first start by giving the ratio of the work cost
probability to the probability of its reverse process for the case maximum possible work was extracted. Then we use that to prove the fluctuation
theorem in the general case. 

Starting with extracting the maximum work in the $\epsilon$-deterministic forward process we
need to find a minimal set of eigenstates such that their total probability is at least $1-\epsilon$.
Removing an $\epsilon$ weight of the initial state amounts to such a smoothing.

To achieve this we $\beta$-order 
the eigenstates of the initial state as explained in the main text.
Then we remove as many eigenstates from the low weight 
end of the spectrum as possible while staying within the $1-\epsilon$ limit. This 
is done by choosing an index $l$, such that
\EQ{
1-\epsilon \geq \sum_j\sum_{i=1}^l k_{i \rightarrow j}\frac{P(E_i)}{g(E-E_i)} \label{geq}
}
and 
\EQ{
1-\epsilon \leq \sum_j \sum_{i=1}^{l+1} k_{i \rightarrow j}\frac{P(E_i)}{g(E-E_i)}. \label{leq}
}

Now we map these eigenvalues of weight $1-\epsilon$ to $W$. For such a
mapping, the ratio of the work gain distribution to its
work cost counterpart is
\begin{widetext}
\EQ{
\nonumber \frac{P(W_{max},\mathcal{P^\epsilon})}{P(-W_{max},\mathcal{P}^{\epsilon,rev})}&=&
\nonumber\frac{\sum_j \Bigg(\Big[\sum_{i=1}^l k_{i \rightarrow j} \frac{P(E_i)}{g(E-E_i)}\Big] + k_{{l+1} \rightarrow j} \frac{P(E_{l+1})}{g(E-E_{l+1})}\Bigg) }{\sum_{i,j}k^{rev}_{j \rightarrow i} \frac{P(E_j)}{g(E-E_j-W)}}\\
&=&\frac{\sum_{i=1}^l g(E-E_i) \frac{P(E_i)}{g(E-E_i)} + \frac{1-\epsilon -\sum_{i=1}^lP(E_i)}{P(E_{l+1})}g(E-E_{l+1}) \frac{P(E_{l+1})}{g(E-E_{l+1})} }{\sum_{i,j}k^{rev}_{j \rightarrow i} \frac{P(E_j)}{g(E-E_j-W)}}\label{epCount}\\
&=& \frac{\sum_{i=1}^l g(E-E_i) \frac{\sum_j k^{rev}_{j \rightarrow i} \frac{P(E_j)}{g(E-E_j-W)}}{g(E-E_i)}+\frac{1-\epsilon -\sum_{i=1}^lP(E_i)}{P(E_{l+1})}g(E-E_{l+1}) \frac{\sum_j k^{rev}_{j \rightarrow {l+1}} \frac{P(E_j)}{g(E-E_j-W)}}{g(E-E_{l+1})} }{\sum_{i,j}k^{rev}_{j \rightarrow i} \frac{P(E_j)}{g(E-E_j-W)}} \label{compBack} \\
&=&\frac{\sum_{i=1}^l g(E-E_i)+\frac{1-\epsilon -\sum_{i=1}^lP(E_i)}{P(E_{l+1})}g(E-E_{l+1})}{\sum_{j} g(E-E_j-W^\epsilon)} \label{Renyi}
}
\end{widetext}
In Eq. (\ref{epCount}) we used Eqs. (\ref{geq}) and (\ref{leq}) to count the number of states with their specific weights that are mapped to the 
work extracted. As opposed to the setting in Lemma \ref{equality}, here th initial state of the reverse process and the final state of the forward
process are not the same. In \cite{Gemmer2015} it was discussed that the final state of the system in an $\epsilon$-deterministic process may not be exactly
a thermal state. To formulate the fluctuation theorem for these processes, we choose the initial state of the backward process to be the same as 
that in the Lemma \ref{equality}, where the backward process starts by the system correlated with the bath, such that the reduced state, describing 
the system is in the Gibb state. Using this, in Eq. (\ref{compBack}), the probability of the system being in energy state $E_{l+1}$ is re-written
in terms of the probabilities in the backward process. 
In the backward process the work $W$ is returned from the battery to the system also via a $\epsilon$-deterministic process. It was shown in 
\cite{AAberg2013} that
for $\epsilon$-deterministic processes, the thermal distribution has a maximum work content of $-\frac{1}{\beta}\ln (1-\epsilon)$. This means 
our reverse process costs this amount less than a deterministic reversal. Therefore we subtract this amount from the work stored in the battery
when calculating the total number of states in the backward process in order to conserve the total energy. 
Hence, defining $W^\epsilon =W+\frac{1}{\beta}\ln (1-\epsilon)$ we have
\EQ{
\sum_i k^{rev}_{j \rightarrow i} = g(E-E_j-W^\epsilon),
}
which gives the denominator in Eq. (\ref{Renyi}). As we see, Eq. (\ref{Renyi}) gives ratio the probability of maximum work extraction to that of 
its reversed process in terms of the number of the eigenstates that are mapped to the work content. However, in an $\epsilon$-deterministic 
scenario we do not always extract the maximum work possible, which is the motivation for having a fluctuation theorem. We would like to have 
the ratio above for any possible work that can be extracted in an $\epsilon$-deterministic process. Let us consider an unknown parameter $\delta$
to characterise a possible work extracted. A work $W^\delta$ is obtained by mapping a $1-\delta$ number of eigenstates in Eq. (\ref{Renyi}) to work.
Since the battery in that case is charged by $W^\delta$ amount, the ratio above generalises to
\begin{widetext}
\EQ{
\nonumber \frac{P(W^\delta,\mathcal{P^\epsilon})}{P(-W^\delta,\mathcal{P}^{\epsilon,rev})}&=&
\nonumber \frac{\Big(\sum_{i=1}^l g(E-E_i)+\frac{1-\epsilon -\sum_{i=1}^lP(E_i)}{P(E_{l+1})}g(E-E_{l+1})\Big)\frac{1-\delta}{1-\epsilon}}{\sum_{j} g(E-E_j-W^\epsilon)} \label{04appft}\\
&=&\exp \{\beta W^\delta +\ln(1-\delta) - D_0^\epsilon (\rho_s||\tau_s)\} \label{finEp}
}
\end{widetext}
Eq. (\ref{finEp}) follows from a similar argument where we derived Eq. (\ref{fin}) from Eq. (\ref{simp}), albeit with smoothing. For the definition of smooth 
Renyi relative entropy see Eq. (\ref{diver}). Smoothing  maximises Renyi relative entropy over all states that are $\epsilon$-close to the initial 
state. As we discussed, in our situation, this is done by removing the eigenstates that have the lowest probability, with the total weight of 
$\epsilon$. The expression in the big brackets in the numerator of Eq. (\ref{Renyi}) corresponds to $g(E)$ times that entropy and is the same expression used in 
\cite{Horodecki2013}.
\end{proof}

\emph{Remark 2 --} For clarity of the proof we have used the notation $W^\delta$ for the work extracted. This corresponds to $W$ in the main text and in particular in the statement of Theorem~\ref{epsilon}.

\emph{Remark 3 --} To formalise the fact that the backward process is the reversal of the forward process in the presence of the epsilon-error, notice that we require that the battery has to be charged 
the amount required to perform the backward process. Although here the global unitary that generates the thermal operation is not unique, one can still decompose any such unitary $U$, as 
a completely positive trace-preserving map on the reservoir and the system, $\mathcal{N}(\cdot)$, and the corresponding translation unitary map $\Gamma(\cdot)$ on the battery as 
$U \rho_{Rsb} U^\dagger= \mathcal{N}_i(\rho_{Rs})\otimes \Gamma_i(\rho_b)$. This means that defining the backward process by the inverse of the forward unitary will send the battery from its charged level to the ground
state. Finally, in order for this process to be a valid backward process, we require that the map $\mathcal{N}$ send the final state of the forward process of Lemma 1 to a state that is close to our initial state. 
Suppose the unitary that has been implemented is one that takes the states that are on the surface of the ball of 
$\epsilon$-close states to $\rho_s$ into $\tau_s$, i.e. the final state achieved in Lemma 1. If the process fails to extract the maximum work, the final state will not exactly be $\tau_s$. Nevertheless, the backward
process, defined by the inverse of the unitary will map the state given by the final state of Lemma 1, to a state that is on the surface of the ball of $\epsilon$-close states to $\rho_s$. This is of course the farthest that the final state of the 
backward process can be from $\rho_s$. Suppose the unitary is implemented, such that mapping does not use all the eigenstates in the range depicted as $\epsilon$ in Fig. (\ref{beta}). In that case the inverse process
will map the final state of Lemma 1 into a state within the ball of $\epsilon$-close states to $\rho_s$. 

\end{document}